\documentclass[11pt]{article}  

\usepackage{multicol}

\usepackage[all]{xy}
\usepackage[pdf]{pstricks}
\usepackage{pst-all}
\usepackage{pstricks-add}

\usepackage{graphicx}

\usepackage{caption}
\usepackage{subcaption}

\usepackage{amssymb}
\usepackage{latexsym}
\usepackage[english]{babel}
\usepackage{amsmath}
\usepackage{makeidx}
\usepackage[utf8]{inputenc}
\usepackage{verbatim}
\usepackage{cite}
\usepackage[pagewise]{lineno}
\usepackage{amsthm}
\usepackage{tikz-cd} 
\usepackage{stmaryrd} 
\usepackage[unicode = true]{hyperref}

\newtheorem{theorem}{Theorem}[section]

\newtheorem{lemma}[theorem]{Lemma}
\newtheorem{proposition}[theorem]{Proposition}

\theoremstyle{definition}
\newtheorem{definition}[theorem]{Definition}
\newtheorem{remark}[theorem]{Remark}
\newtheorem{example}{Example}

\newcommand{\R}{\ensuremath{\mathbb{R}}}

\newcommand{\D}{\ensuremath{\mathcal{D}}}

\newcommand{\F}{\ensuremath{\mathbb{F}}}

\newcommand{\dd}{d}
\newcommand{\Reeb}{R}

\DeclareMathOperator{\Id}{Id}

\begin{document}
	\title{The geometry of some thermodynamic systems}
	
	\author{
		{\bf\large Alexandre Anahory Simoes}\hspace{2mm}
		\vspace{1mm}\\
		{\it\small Instituto de Ciencias Matemáticas (CSIC-UAM-UC3M-UCM)}\\
		{\it\small Calle Nicolás Cabrera, 13-15, Campus Cantoblanco, UAM}, {\it\small 28049 Madrid, Spain}\\
		\vspace{2mm}\\
		{\bf\large David Mart{\'\i}n de Diego}\hspace{2mm}
		\vspace{1mm}\\
		{\it\small Instituto de Ciencias Matem\'aticas (CSIC-UAM-UC3M-UCM) }\\
		{\it\small Calle Nicol\'as Cabrera, 13-15, Campus Cantoblanco, UAM}, {\it\small 28049 Madrid, Spain}\\
		\vspace{2mm}\\
		{\bf\large Manuel Lainz Valcázar}\hspace{2mm}
		\vspace{1mm}\\
		{\it\small Instituto de Ciencias Matem\'aticas (CSIC-UAM-UC3M-UCM) }\\
		{\it\small Calle Nicol\'as Cabrera, 13-15, Campus Cantoblanco, UAM}, {\it\small 28049 Madrid, Spain}\\
		\vspace{2mm}\\
		{\bf\large Manuel de Le\'on}\hspace{2mm}
		\vspace{1mm}\\
		{\it\small Instituto de Ciencias Matem\'aticas and Real Academia Espa\text{\~{n}}ola de Ciencias}\\
		{\it\small Calle Nicol\'as Cabrera, 13-15, Campus Cantoblanco, UAM}, {\it\small 28049 Madrid, Spain}\\
	}

	\maketitle
	
	\begin{abstract}
In this article, we continue the program started in~\cite{termo1} of exploring an important class of thermodynamic systems from a geometric point of view. In order to model the time evolution of systems verifying the two laws of thermodynamics, we show that the notion of evolution vector field is adequate to appropriately describe such systems. Our formulation naturally arises from the introduction of a skew-symmetric bracket to which numerical methods based on discrete gradients fit nicely. Moreover, we study the corresponding Lagrangian and Hamiltonian formalism, discussing the fundamental principles from which the equations are derived. An important class of systems that is naturally covered by our formalism are composed thermodynamic systems, which are described by at least two thermal variables and exchange heat between its components.	\end{abstract} 
	
	\let\thefootnote\relax\footnote{\noindent AMS {\it Mathematics Subject Classification (2010)}. Primary 37J55; Secondary  37D35, 70G45, 80M25.}

\section{Introduction}
\label{sec:1}

In this paper, we continue the differential geometric study initiated in our paper about the  evolution vector field associated with a contact system~\cite{termo1}. However, when we analyze more complex thermodynamic examples, we will soon realize that it is often necessary to consider alternative geometrical structures from those of contact geometry and/or combine them properly.
	
Our approach differs from others in  the previous literature, where different authors introduce simultaneously a skew-symmetric and a symmetric bracket with combined properties that allow the two laws of thermodynamics to be satisfied. This is the case, for instance, of metriplectic structures (see~\cite{Kaufman, morrison} and references therein) or also the so-called single generation formalism~\cite{Ed-Ber, Ed-Ber-2}. Alternatively, other authors such as Gay-Balmaz and Yoshimura introduce in~\cite{Gay-Balmaz2017, Gay-Balmaz2019} a ``variational principle'' for the description of thermodynamic systems by means of	phenomenological and variational constraints.
	
The application of contact geometry (see~\cite{Hermann, Arn1990}) to model thermodynamics is suggested by Gibbs’ fundamental relation, which relates extensive and intensive variables defining the state of thermodynamic systems. Contact geometry is the  odd  dimensional counterpart of symplectic geometry~\cite{Godbillon1969, marle}. From this point of view, the flow of the restriction of the contact vector field to a Legendre submanifold of an appropriate contact manifold is interpreted as a thermodynamic process~\cite{Mruga,Mrugala1991}.
The most familiar symplectic framework can be obtained by {\em symplectification} of the contact manifold, obtaining a symplectic manifold with an additional structure of homogeneity bringing together energy and entropy representations (see also~\cite{BV2001,VdSM2018}).
	
In the present paper, after quickly reviewing contact geometry in Section \ref{sec2}, we focus not on the dynamics derived from the contact vector field as in~\cite{Bravetti2017,Bravetti2018, deLeon2018} but instead we will choose a different vector field: the so-called evolution vector field. We remark that the evolution vector field is defined exclusively in terms of a bi-vector or, equivalently, a skew-symmetric bracket of functions. This fact allows the use of discrete gradient methods~\cite{GONZ,ITOH,termo1} to numerically simulate the dynamics of thermodynamic processes. In~\cite{termo1}, we show that the evolution vector field preserves the Hamiltonian, which means that it models an isolated system. 
Moreover, we show that this vector field is tangent to the kernel of the contact form, which corresponds to the first law of thermodynamics. 

As a novelty, we state in Section \ref{sec:Lagrangian} a ``variational principle'' for the evolution dynamics in a Lagrangian framework, which is a generalization of Chetaev's principle (cf.~\cite{CILM2004}).

In Section \ref{sec4}, we briefly review some basic principles of Thermodynamics and how the evolution vector field models isolated simple thermodynamic systems with friction (see~\cite{termo1}).

In the final part of our paper, we will show in Section \ref{sec5} that the evolution vector field or Hamiltonian vector field associated to a skew-symmetric bracket is also useful to describe more complex systems outside the scope of contact geometry. For instance, we study as an example a system composed of at least two subsystems exchanging heat with each other~\cite{Romero2010, PGR2017}.

After that, in Section~\ref{sec6} we present numerical methods that have the property of satisfying the two laws of thermodynamics. To illustrate this point, we run a simulation for each class of systems addressed in the paper.

Finally, in the last section, we conclude pointing out some possible directions for further research.

\section{Contact geometry}\label{sec2}
In this section, we recall the facts about contact geometry which will be necessary for developing our formalism~\cite{Godbillon1969, marle,deLeon2018}. We will define the Hamiltonian and the evolution vector fields and compare its properties.
	
	Let $(M, \eta)$ be a contact manifold. That is, $M$ is a $(2n+1)$-dimensional manifold and $\eta$ is a  1-form that satisfies $\eta\wedge (d\eta)^n\not=0$ at every point.
	
	The Reeb vector field $R\in {\mathfrak X}(M)$ is the unique vector field satisfying 
	\[
	i_R\eta=1 \quad \hbox{and}\quad  i_Rd\eta=0\; .
	\]
	The contact form $\eta$ induces an isomorphism of $C^{\infty}(M, \R)$ modules by
		\[
	\begin{array}{rrcl}
\flat:& {\mathfrak X}(M)&\longrightarrow& \Omega^1(M)\\
	& X&\longmapsto& i_Xd\eta +\eta(X)\eta
	\end{array}
	\]
	Observe that $\flat^{-1}(\eta)=R$. 

	Using the generalized Darboux theorem,  we have canonical coordinates $(q^i, p_i, S)$, $1\leq i\leq n$ in a neighborhood of every point $x\in M$, such that the contact 1-form $\eta$ and the Reeb vector field  are: 
	\[
	\eta=dS-p_i\; dq^i \qquad \hbox{and} \qquad R=\frac{\partial}{\partial S}\; .
	\]

	The contact structure also provides a Whitney sum decomposition of the tangent bundle $TM =  \ker \eta \oplus \langle \Reeb \rangle$, with projectors
	\begin{equation}\label{contact_projectors}
		\mathcal{P} = \Id - \Reeb \otimes \eta \; \text{ and } \; \mathcal{Q} = \Reeb \otimes \eta,
	\end{equation}
	 onto $\ker \eta$ and $\langle \Reeb \rangle$, respectively.
	 
		An important example of contact manifold is the extended cotangent bundle $T^*Q\times \R$, where  $Q$ is $n$-dimensional manifold, which is naturally equipped with the following contact form
	\[
	\eta_Q =pr_2^* (dS)-pr_1^* (\theta_Q)\equiv dS-\theta_Q
	\]
	where $pr_1: T^*Q\times \R\rightarrow T^*Q$ and $pr_2: T^*Q\times \R\rightarrow \R$  are the canonical projections and $\theta_Q$ is the Liouville 1-form on the cotangent bundle
	defined by
	\[
	\theta_Q (X_{\mu_q})=\langle \mu_q, T_{\mu_{q}}\pi_Q X_{\mu_q}\rangle
	\]
	being $X_{\mu_q}\in T_{\mu_q} T^*Q$.  
Taking bundle  coordinates $(q^i,p_i)$ on $T^*Q$ we have that 
$\eta=dS-p_idq^i$.

On this situation, we notice that $\omega =  {\dd \eta}\vert_{\ker \eta}$, then $(\ker \eta, \omega)$ is a symplectic vector bundle over $T^*Q \times \R$. Hence, there is a unique vector field $\Delta_Q$ tangent to $\ker \eta$ satisfying $i_{\Delta_Q} \omega = \theta_Q \vert_{\ker \eta}$. We call $\Delta_Q$ the \emph{Liouville vector field} and, in Darboux coordinates, it is given by
	\begin{equation}
		\Delta_Q = p_i \frac{\partial }{\partial p_i}.
	\end{equation}

	We remark that $(\dd S, -\dd \theta_Q)$ is a cosymplectic structure which can be naturally constructed on $T^*Q \times \R$. Indeed, $\dd S$ and $-\dd \theta_Q$ are exact and $\dd S \wedge (\dd \theta)^n$ is a volume form. However, the dynamics arising from this cosymplectic structure are different from the ones arising from the contact structure~\cite{deLeon2019}.

	\subsection{The Jacobi structure of a contact manifold}
	Contact manifolds have an associated Jacobi structure. Indeed, we define the bivector $\Lambda$ as  
\begin{equation}\label{Lambda:intrinsic}
	\Lambda(\alpha, \beta)=-d\eta(\flat^{-1}(\alpha), \flat^{-1}(\beta)), \qquad \alpha, \beta \in \Omega^1(M)\; .
\end{equation}
In local coordinates, the bivector reads as
\begin{equation}\label{Lambda:coordinates}
	\Lambda=\frac{\partial}{\partial p_i}\wedge \left(\frac{\partial}{\partial q^i}+p_i\frac{\partial}{\partial S}\right).
\end{equation}
So that the pair $(\Lambda, E=-R)$ satisfies
\[
[\Lambda, \Lambda]=2E\wedge \Lambda \quad \hbox{and} \quad [\Lambda, E]=0\; .
\]
We define the morphism of $C^{\infty}(M, \R)$-modules $$\sharp_{\Lambda}: \Omega^1(M)\rightarrow {\mathfrak X}(M)$$ by $\langle \beta, \sharp_\Lambda(\alpha)\rangle=\Lambda(\alpha, \beta)$ with  $\alpha, \beta \in \Omega^1(M)$. 

	From this  Jacobi structure we can define the Jacobi bracket as follows: 
	\[
	\{f, g\}=\Lambda(df, dg)+f E(g)-g E(f), \quad f, g\in C^{\infty}(M, \R)
	\]
	The mapping  $\{\;  ,\;  \}: C^{\infty}(M, \R) \times  C^{\infty}(M, \R) \longrightarrow C^{\infty}(M, \R)$ is bilinear, skew-symmetric and satisfies the Jacobi’s identity but, in general, it does not satisfy the Leibniz rule; this last property is replaced by a weaker condition: 
	\[
	\hbox{Supp} \ \{f, g\}\subset 	\hbox{Supp} \ f\cap \hbox{Supp} \ g\; .
	\]
	This last condition is equivalent to the generalized Leibniz rule
	\begin{equation}
	    \{f, gh\} = g\{f, h\} + h\{f, g\} + ghE(h),
	\end{equation}
	In this sense, this bracket generalizes the well-known Poisson brackets. Indeed, a Poisson manifold is a particular case of Jacobi manifold in which $E=0$. 
	
	In local coordinates, the bracket is given by
	\begin{eqnarray*}
	\{f,g\} = \frac{\partial f}{\partial q^i}\frac{\partial g}{\partial p_i}
	-\frac{\partial f}{\partial S}\left(p_i\frac{\partial g}{\partial p_i}-g\right)+\frac{\partial g}{\partial S}\left(p_i\frac{\partial f}{\partial p_i}-f\right)
	\end{eqnarray*}
	
	We can define another bracket, the \textit{Cartan bracket}, by only using the bivector of the Jacobi structure. This bracket is bilinear, skew-symmetric and satisfies the Leibniz rule, but does not obey the Jacobi identity
	\begin{eqnarray*}
		[f, g]&=&\Lambda (df, dg)\\
		&=&\frac{\partial f}{\partial p_i}\frac{\partial g}{\partial q^i}-
		\frac{\partial f}{\partial q^i}\frac{\partial g}{\partial p_i}
		-\frac{\partial f}{\partial S}\left(p_i\frac{\partial g}{\partial p_i}\right)+\frac{\partial g}{\partial S}\left(p_i\frac{\partial f}{\partial p_i}\right)
	\end{eqnarray*}
		
The third bracket can be defined from the following bivector
\[
\Lambda_0=\Lambda+ \Reeb \wedge \Delta_Q
\]
which is Poisson, that is, $[\Lambda_0, \Lambda_0]=0$. 
In coordinates, 
\[
\Lambda_0=\frac{\partial}{\partial p_i}\wedge \frac{\partial}{\partial q^i}
\]
is like the canonical Poisson bracket on $T^*Q$ but now applied to functions on $T^*Q\times \R$.  In fact, this bracket the Jacobi bracket related to the cosymplectic structure $(\dd S,-\dd \theta_Q)$~\cite{deLeon2017}.

On the case that $M = T^*Q \times \R$, the Cartan bracket can be rewritten in terms of the Poisson bracket induced by $\Lambda_0$ and an extra term describing the thermodynamic behaviour. That is, 
\[
[f, g]=\{f, g\}_{\Lambda_0}-\frac{\partial f}{\partial S}\Delta_Q g+\frac{\partial g}{\partial S}\Delta_Q f
\]
We will denote by 
\[
\{
f, g
\}_{\Delta_Q}=\frac{\partial g}{\partial S}\Delta_Q f-\frac{\partial f}{\partial S}\Delta_Q g
\]
	then the Cartan bracket is written as in the single generation formalism~\cite{Ed-Ber, Ed-Ber-2}  as
	\begin{equation}\label{eee}
	[f, g]=\{f, g\}_{\Lambda_0}+\{
	f, g
	\}_{\Delta_Q}
	\end{equation}

\subsection{Hamiltonian and evolution vector fields}
Given a function $f: M \to \R$ on a contact manifold $(M, \eta)$ we define the following vector fields
\begin{itemize}
\item {\bf Hamiltonian or contact vector field} $X_f$ defined by
\[
X_f=\sharp_{\Lambda} (df)-fR
\]
equivalently, $X_f$ is the unique vector field such that
\[
\flat(X_f)=df-(R(f)+f)\, \eta\; .
\]

In canonical coordinates:
\[
X_f=\frac{\partial f}{\partial p_i}\frac{\partial}{\partial q^i}
-\left(\frac{\partial f}{\partial q^i}+p_i\frac{\partial f}{\partial S}\right)
\frac{\partial}{\partial p_i}+
\left(p_i\frac{\partial f}{\partial p_i}-f\right)\frac{\partial}{\partial S}
\]

\item {\bf The evolution or horizontal  vector field}
\[
{\mathcal E}_f=\sharp_{\Lambda} (df)=X_f+fR
\]
or
\[
\flat({\mathcal E}_f)=df-R(f)\, \eta\; .
\]
In canonical coordinates:
\begin{equation}\label{eq:evolutionvf}
	{\mathcal E}_f=\frac{\partial f}{\partial p_i}\frac{\partial}{\partial q^i}
-\left(\frac{\partial f}{\partial q^i}+p_i\frac{\partial f}{\partial S}\right)
\frac{\partial}{\partial p_i}+
p_i\frac{\partial f}{\partial p_i}\frac{\partial}{\partial S}
\end{equation}
\end{itemize}
	
	Now we will compare some properties of the Hamiltonian and evolution vector fields. Proofs can be found in~\cite{deLeon2018,termo1}
	\begin{proposition}\label{Hamiltonian_vs_evolution}
		The evolution vector field preserves the energy $f$, but the Hamiltonian vector field dissipates it:
		\begin{align*}
			X_f (f) &= - \Reeb(f) f,\\
			\mathcal{E}_f (f) &= 0.
		\end{align*}

		The Hamiltonian vector field preserves $\ker \eta$ (indeed, the flow of $X_f$ changes $\eta$ by a conformal factor), but the evolution vector field does not
		\begin{align*}
			\mathcal{L}_{X_f} \eta &=  -\Reeb(f) \eta, \\
			\mathcal{L}_{\mathcal{E}_f} \eta &= -\Reeb(f) \eta + \dd f.
		\end{align*}

		The evolution vector field is everywhere contained in $\ker \eta$, while the Hamiltonian vector field is only in $\ker \eta$  at the zero level set of the energy
		\begin{align*}
			i_{X_f} \eta &= -f, \\
			i_{\mathcal{E}_f} \eta &= 0.	
		\end{align*}
	\end{proposition}
	
	We compare the dynamics of the Hamiltonian~\cite{BLMP} and evolution~\cite{termo1} vector fields
	\begin{theorem}
		Let ${\mathcal L}_c(f)=f^{-1}(c)$ be a level set of $f: M\rightarrow \R$ where $c\in \R$. We assume that ${\mathcal L}_c(f)\not=0$ and $R(f)(x)\not= 0$ for all $x\in {\mathcal L}_c(f)$. 
		
		Then The 2-form $\omega_c\in \Omega^2({\mathcal L}_c(f))$ defined by
		\[
		\omega_c=-d i_{c}^*\eta
		\]
		is an exact symplectic structure, where
		 $i_c: {\mathcal L}_c f\hookrightarrow M$ is  the canonical inclusion 
		  Let $\Delta_c$ is the Liouville vector field of $T{\mathcal L}_c(f)$, given by 
	\[
	i_{\Delta_c}\omega_c=i_{c}^*\eta
	\]

	\begin{enumerate}
		\item  For the Hamiltonian vector field $X_f$:
		\begin{itemize}
			\item When $c \neq 0$, $X_f$ is the Reeb vector field of the contact form $\tilde{\eta} = \eta/f$.
			\begin{equation*}
				i_{X_f} \tilde{\eta}=1, \quad  i_{X_f} \tilde{\eta}=0. 
			\end{equation*}
			\item When $c=0$, the dynamics of $X_f$ is a reparametrization of the dynamics of the Liouville vector field.
			\begin{equation*}
				{X}_f\big|_{{\mathcal L}_0 (f)}=R(f)\big|_{{\mathcal L}_c(f)} \Delta_0
			\end{equation*}
		\end{itemize}
		\item The dynamics of $\mathcal{E}_f$ is a reparametrization of the dynamics of the Liouville vector field at any constant energy hypersurface.
		\begin{equation*}
			{\mathcal E}_{f}\big|_{{\mathcal L}_c (f)}=R(f)\big|_{{\mathcal L}_c(f)} \Delta_c
		\end{equation*}
	\end{enumerate}
\end{theorem}
	
\section{The Lagrangian formalism}\label{sec:Lagrangian}

\subsection{The geometric setting}

Let $L : TQ \times \R \longrightarrow \R$ be a regular contact Lagrangian function (see~\cite{deLeon2019,deLeon2020}). As before, let us introduce coordinates on $TQ \times \R$, denoted by $(q^i, \dot{q}^i, S)$, where $(q^i)$ are coordinates in $Q$, $(q^i, \dot{q}^i)$ are the induced bundle coordinates in $TQ$	and $S$ is a global coordinate in $\R$.

Given a Lagrangian function $L$, using the canonical endomorphism ${\mathbf S}$ on $TQ$ locally defined by
$$	{\mathbf S}= d q^i \otimes \frac{\partial}{\partial \dot{q}^i},	$$
one can construct a 1-form $\lambda_L$ on $TQ\times \R$ given by
$$	\lambda_L = {\mathbf S}^* (dL)	$$	
where now ${\mathbf S}$ and ${\mathbf S}^*$ are the natural extensions of ${\mathbf S}$ and its adjoint operator ${\mathbf S}^*$ to $TQ \times \R$~\cite{deLeon1987}.

Therefore, we have that
$$	\lambda_L = \frac{\partial L}{\partial \dot{q}^i} \, dq^i.	$$

Now, the 1-form on $TQ\times \R$ given by $\eta_L = dS -\lambda_L$ or, in local coordinates, by	
$$	\eta_L = dS - \frac{\partial L}{\partial \dot{q}^i} \, dq^i	$$
is a contact form on $TQ \times \R$ if and only if $L$ is regular; indeed, if $L$ is regular, then we may prove that	
$	\eta_L \wedge (d\eta_L)^n \not= 0$,
and the converse is also true. 

The corresponding Reeb vector field is given in local coordinates by
$$	{\mathcal R}_L = \frac{\partial}{\partial S} - W^{ij} \frac{\partial^2 L}{\partial \dot{q}^j \partial S} \, \frac{\partial}{\partial \dot{q}^i} ,$$
where $(W^{ij})$ is the inverse matrix of the Hessian $(W_{ij})$ with
\begin{equation}
    W_{ij} = \frac{\partial^2 L}{\partial \dot{q}^i \partial \dot{q}^j}.
\end{equation}

The energy of the system is defined by 
$$E_L = \mathbf{\Delta} (L) - L	$$
where $\mathbf{\Delta} = \dot{q}^i \, \frac{\partial}{\partial \dot{q}^i}$ is the natural extension of the Liouville vector field on $TQ$ to $TQ \times \R$. Therefore, in local coordinates we have that	
$$E_L = \dot{q}^i \, \frac{\partial L}{\partial \dot{q}^i} - L.$$

Denote by $\flat_L : T(TQ \times \R) \longrightarrow T^* (TQ \times \R)$	the vector bundle isomorphism given by	
$$\flat_L (v) = i_v (d\eta_L) + (i_v \eta_L) \, \eta_L$$
where $\eta_L$ is the contact form on $TQ \times \R$ previously defined. We shall denote its inverse isomorphism by $\sharp_L = (\flat_L)^{-1}$.

Let ${\xi}_L$ be the unique vector field satisfying the equation	
\begin{equation}\label{clagrangian1}
\flat_L ({\xi}_L) = dE_L - (\mathcal R_L E_L + E_L) \, \eta_L.
\end{equation}	

A direct computation from eq. (\ref{clagrangian1}) shows that if $(q^i(t), \dot{q}^i(t), S(t))$ is an integral curve of ${\xi}_L$, then it satisfies the generalized Euler-Lagrange equations considered by G. Herglotz in 1930:
\begin{equation}\label{clagrangian4}
\begin{split}
& \frac{d}{dt} \left(\frac{\partial L}{\partial \dot{q}^i}\right) - \frac{\partial L}{\partial q^i} = \frac{\partial L}{\partial \dot{q}^i} \frac{\partial L}{\partial S}\; ,\\
& \dot{S}=L(q^i, \dot{q}^i, S)\; .
\end{split}
\end{equation}

Now, given a regular Lagrangian function $L$, we may define the bi-vector $\Lambda_{L}$ on $TQ\times\R$ as in \eqref{Lambda:intrinsic} associated to the contact form $\eta_{L}$. That is, 
\begin{equation}\label{Lambda:intrinsic-1}
\Lambda_L(\alpha, \beta)=-d\eta_L(\flat_L^{-1}(\alpha), \flat_L^{-1}(\beta)), \qquad \alpha, \beta \in \Omega^1(TQ\times \R)\; .
\end{equation}

If $(q^i(t), \dot{q}^i(t), S(t))$ is an integral curve of the evolution vector field ${\Gamma}_{L}$ associated to the contact form $\eta_{L}$, which is the SODE vector field defined by
\begin{equation*}
{\Gamma}_{L}=\sharp_{\Lambda_{L}}(dE_{L}) \hbox{   or   }  	\flat_L ({\Gamma}_L) = dE_L - (\mathcal R_L E_L) \, \eta_L\; ,
\end{equation*}
then the curve satisfies the thermodynamic Herglotz equations
\begin{equation}\label{Herglotz:thermo}
\begin{split}
& \frac{d}{dt} \left(\frac{\partial L}{\partial \dot{q}^i}\right) - \frac{\partial L}{\partial q^i} = \frac{\partial L}{\partial \dot{q}^i} \frac{\partial L}{\partial S}. \\
& \dot{S}=\dot{q}^i\frac{\partial L}{\partial \dot{q}^i}.
\end{split}
\end{equation}

Moreover, if $H$ is the Hamiltonian function defined by $H=E_{L}\circ (\F L)^{-1}$, where $\F L:TQ\times \R\rightarrow T^{*}Q\times \R$ is the Legendre transform, then the evolution vector field ${\mathcal E}_{H}$ associated to $H$ is $\F L$-related to ${\Gamma}_{L}$.

\subsection{Generalized Chetaev principle}

In order to present a pseudo-variational principle, we will first formulate a variational problem for systems subjected to a particular non-standard type of nonholonomic constraints.

Let $M$ be a manifold and let $\eta$ be a semibasic 1-form on $TM$. Locally, if $(q^{i})$ are coordinates on $M$ and $(q^{i},\dot{q}^{i})$ are natural bundle coordinates on $TM$, we have that
\begin{equation*}
\eta(q,\dot{q})=\eta_{k}(q,\dot{q}) d q^{k}.
\end{equation*}
To every semibasic 1-form, we may associate in a canonical way a force map $F:TM\rightarrow T^{*}M$ by the formula
\begin{equation*}
\langle F(q,\dot{q}), T\tau_{M}(X) \rangle = \langle \eta (q,\dot{q}),X \rangle, \quad X\in T_{(q,\dot{q})}(TM),
\end{equation*}
where $\tau_{M}:TM\rightarrow M$ is the canonical tangent bundle projection. Locally, we have that
\begin{equation*}
F(q,\dot{q})=\eta_{k}(q,\dot{q}) d q^{k}.
\end{equation*}

Now, given a vector $w_{q}\in T_{q} M$, consider the set
\begin{equation*}
\mathcal{D}_{w_{q}}=\{ v_{q}\in T_{q}M \ | \ \langle F(w_{q}), v_{q} \rangle=0 \} \subseteq T_{q}M.
\end{equation*}

A \textit{semibasic nonholonomic constraint} is specified by a semibasic 1-form $\eta$ or, equivalently, by a smooth assignment
\begin{equation*}
w_{q} \mapsto \D_{w_{q}}.
\end{equation*}
Then, a trajectory $c:I\rightarrow M$ is said to satisfy the constraint if $\langle F(\dot{c}),\dot{c} \rangle=0$ or, equivalently, $\dot{c}\in \D_{\dot{c}}$.

\begin{definition}[Generalized Chetaev principle]
	Given a Lagrangian function $L:TM\rightarrow \R$ and a semibasic nonholonomic constraint $\eta$, a trajectory $c:I\rightarrow M$ satisfies the generalized Chetaev principle if
	\begin{equation}
	\delta \int L dt = 0 \quad \text{and} \quad \dot{c}\in \D_{\dot{c}},
	\end{equation}
	among variations with fixed endpoints satisfying $\delta c\in \D_{\dot{c}}$.
\end{definition}

Similarly to what happens with variational calculus we may obtain a set of equations as necessary and sufficient conditions to find trajectories satisfying the generalized Chetaev principle.

\begin{lemma}
	Given a Lagrangian function $L:TM\rightarrow \R$ and a semibasic nonholonomic constraint $\eta$, a trajectory $c:I\rightarrow M$ satisfies the generalized Chetaev principle if and only if ot satisfies the equations
	\begin{equation}\label{GCE}
	\begin{split}
	& \frac{d}{d t}\left( \frac{\partial L}{\partial \dot{q}^{i}} \right)-\frac{\partial L}{\partial q^{i}}=\lambda \eta_{i} \\
	& \eta_{i}\dot{q}^{i}=0.
	\end{split}
	\end{equation}
\end{lemma}

\begin{proof}
	Using the standard arguments from calculus of variations, we deduce that
	\begin{equation*}
		\delta \int L dt = \int \left[ \frac{\partial L}{\partial q^{i}}-\frac{d}{d t}\left( \frac{\partial L}{\partial \dot{q}^{i}} \right) \right] \delta q \ d t,
	\end{equation*}
	where we used integration by parts and the fact that the infinitesimal variations vanish at the endpoints. Then, in order that the integrand vanishes for arbitrary variations satisfying $\delta c\in \D_{\dot{c}}$ we must have that
	\begin{equation*}
		\frac{\partial L}{\partial q^{i}}-\frac{d}{d t}\left( \frac{\partial L}{\partial \dot{q}^{i}}\right) \in \D_{\dot{c}}^{o}.
	\end{equation*}
	But $\D_{\dot{c}}^{o}=\text{span}\{F(\dot{c})\}\subseteq T_{c}^{*}M$. Thus, the expression above must be a saclar multiple of the co-vector $F(\dot{c})$. Putting this fact together with the constraint $\dot{c}\in \D_{\dot{c}}$ we obtain equations \eqref{GCE}.
\end{proof}

The variational principle we have just stated is satisfied by the Lagrangian evolution vector field. More precisely, the integral curves of the Lagrangian evolution vector field satisfy a nonholonomic variational principle with nonlinear constraints. Indeed, the solutions of the equations of motion are critical points of the action with a condition of tangency to the contact distribution. This nonlinear nonholonomic principle is exactly the one described above and it is similar to the one introduced in~\cite{Gay-Balmaz2017}.

To observe this clearly, take $M$ to be the manifold $Q\times \R$ and consider the extended Lagrangian function $\hat{L}:T(Q\times \R)\rightarrow \R$ defined by
\begin{equation*}
	\hat{L}(v_{q},\zeta_{S})=L(v_{q},S), \quad (v_{q},\zeta_{S})\in T_{(q,S)}(Q\times \R),
\end{equation*}
where $L:TQ\times \R\rightarrow\R$ is a contact Lagrangian function. Take the pullback of the Lagrangian 1-form $\eta_L$ to $T(Q \times \mathbb{R})$, which we will also denote by $\eta_{L}$ and it is a semibasic 1-form. Let $(q^{i})$ be coordinates on $Q$, $(q^{i},\dot{q}^{i})$ natural tangent bundle coordinates on $TQ$ and $(q^{i},S,\dot{q}^{i},\dot{S})$ be coordinates on $T(Q\times \R)$. The local expression of the semibasic 1-form $\eta_{L}$ is
\begin{equation}
	\eta_{L}(q,S,\dot{q},\dot{S})=d S - \frac{\partial L}{\partial \dot{q}^i} d q^i.
\end{equation}

As before, we may associate to $\eta_{L}$ a force map given by
\begin{align*}
F_L: T(Q \times \mathbb{R}) &\to T^*(Q \times \mathbb{R})\\
(q, S, \dot{q}, \dot{S}) &\mapsto d S - \frac{\partial L}{\partial \dot{q}^i} d q^i.
\end{align*}

The nonholonomic constraint is determined at each point $(v_{q},\zeta_{S})\in T_{(q,S)}(Q\times \R)$ by those vectors $(\gamma_{q},\xi_{S})\in T_{(q,S)}(Q\times \R)$ such that 
\begin{equation*}
	\langle F_L(v_{q},\zeta_{S}),(\gamma_{q},\xi_{S})  \rangle=0.
\end{equation*}

\begin{proposition}
	Given a contact Lagrangian $L:TQ\times \R\rightarrow \R$, a curve $(q(t),\dot{q}(t),S(t))$ is an integral curve of the evolution vector field $\Gamma_{L}$ if and only if the associated curve $(q(t),S(t),\dot{q}(t),\dot{S}(t))$ satisfies the generalized Chetaev principle for the extended Lagrangian $\hat{L}$ with nonholonomic constraints determined by the semibasic 1-form $\eta_{L}$.
\end{proposition}

\begin{proof}
	We will compare the equations satisfied by integral curves of the evolution vector field with the ones satisfied by the solutions of the corresponding generalized Chetaev principle.
	
	Indeed, the solution of the generalized Chetaev principle satisfy
	\begin{subequations}
		\begin{align}
		\frac{d}{dt} \left(\frac{\partial L}{\partial \dot{q}^i}\right) - \frac{\partial L}{\partial q^i} &= - \lambda \frac{\partial L}{\partial \dot{q}^i}. \\
		\frac{d}{dt} \left(\frac{\partial L}{\partial \dot{S}}\right) - \frac{\partial L}{\partial S} &= \lambda.
		\end{align} 
	\end{subequations}

Since $L$ does not depend on $\dot{S}$, the last equation is reduced to
\begin{equation}
\frac{\partial L}{\partial S} = -\lambda,
\end{equation}
So we retrieve the equations for the Lagrangian evolution vector field~\eqref{Herglotz:thermo} by adding the constraint equation
\begin{equation}
    \dot{S} = \frac{\partial L}{\partial \dot{q}^i} \dot{q}^i
\end{equation}
\end{proof}

\section{The evolution vector field and simple mechanical systems with friction}\label{sec4}

	\subsection{The laws of thermodynamics}
	In this section we will apply the evolution vector field to the description of simple mechanical vector field with friction. These systems can be described through one scalar thermal variable (in our formulation it will be the entropy) and finitely many mechanical variables (positions and velocities, in the Lagrangian formalism, or positions and momenta in Hamiltonian formalism). Furthermore, our system will be isolated: there will be no transfer of any form of matter or energy with its surroundings.

	The dynamics of the system will be described through a Lagrangian 
	\[
	\begin{array}{rrcl}
	L:& TQ\times \R&\longrightarrow& \R,\\
	   & (q^i, \dot{q}^i , S)&\longmapsto& L(q^i, \dot{q}^i, S),
	\end{array}
	\]
	or a Hamiltonian
	\[
	\begin{array}{rrcl}
	H:& T^*Q\times \R&\longrightarrow& \R,\\
	   & (q^i, p_i, S)&\longmapsto& H(q^i, p_i, S),
	\end{array}
	\]
	which, assuming that the Lagrangian is regular, both systems are connected through the Legendre transform, as explained in Section~\ref{sec:Lagrangian}.

	The integral curves of the evolution vector will describe the trajectories of the system. As we will see, these curves will satisfy the first law of thermodynamics for an isolated system. 
	
	The thermodynamic space of a Lagrangian system is naturally equipped with two linearly independent one forms: the work $\delta \mathcal{W}$ and the heat $\delta \mathcal{Q}$ one-forms. The energy of the system is given by an energy function $H:T^* Q \times \mathbb{R} \to \R$. For a closed system (one that does not exchange matter, but may exchange energy), the first law can be written as follows. Along any process, $\chi$,
	\begin{equation}
		d H = \delta \mathcal{Q} - \delta \mathcal{W}.
	\end{equation}
	Since our system is simple, the form $\delta \mathcal{Q}$ needs to have rank one. Therefore, it can be written as
	\begin{equation}
		\delta \mathcal{Q} = T \dd S,
	\end{equation}
	for some functions $T$ and $S$ (which are the temperature and the entropy). Furthermore, $\delta \mathcal{W}$ can be written locally as follows
	\begin{equation}
		\delta \mathcal{W} = P_i \dd q^i,
	\end{equation}
	which as many functions $P_i$ and $q^i$ as the rank of $\delta \mathcal{W}$. Moreover, $q^i$ and $S$ are functionally independent. In the physical interpretation, $P_i$ is the pressure. Hence, along $\chi$, the following is satisfied
	\begin{equation}
		d H = T d S - P_i d q^i.
	\end{equation}
	From this, by contracting with $\partial/\partial S$, we obtain the relationship
	\begin{equation}
		T = \frac{\partial H}{\partial S}.
	\end{equation}
	Furthermore, for an isolated system, the energy must be constant along $\chi$. Hence, we must have the relationship
	\begin{equation}
		0 = T d S - P_i d q^i,
	\end{equation}
	or, dividing by the temperature, and identifying
	\begin{equation}
		p_i = P_i/T,
	\end{equation}
	we obtain
	\begin{equation}
		\eta_Q =  d S - p_i d q^i = 0
	\end{equation}
	Hence, $\chi$ satisfies the first law of thermodynamics for an isolated system if and only if the energy is constant along $\chi$ and $\chi$ is tangent to $\ker \eta_Q$.

	From this comment and Proposition~\ref{Hamiltonian_vs_evolution} we can extract the following conclusion. 
\begin{proposition}\label{prop:law1}
	The integral curves of ${\mathcal E}_H$ describe an isolated system, that is
	$$\frac{d H}{dt}=0.$$ 
	Moreover, the time evolution of the entropy is locally given by
	$$\frac{\dd S}{\dd t} = p_i \frac{\dd q^i}{\dd t},$$
	which is exactly the first law of thermodynamics with $p_i = P_i/T$.
\end{proposition}

\begin{remark}
Note that the first law of thermodynamics for an isolated system may be geometrically written as a tangency condition, that is,
$$i_{{\mathcal E}_H} \eta = 0.$$
\end{remark}

The second law of thermodynamics follows from the expression of the evolution vector fields \eqref{eq:evolutionvf}, \eqref{Herglotz:thermo} and depends on the choice of Hamiltonian function.
\begin{proposition}
	The integral curves of ${\mathcal E}_H$ (respectively {$\Gamma_L$}) satisfy the Second law of thermodynamics, that is,
	\begin{equation}
		\frac{\dd S}{\dd t} \geq 0
	\end{equation}
	 if and only if $\Delta_Q(H) \geq 0$ (respectively, $\mathbf{\Delta}(L) \geq 0$).
\end{proposition}

\subsection{An example}
Let the Hamiltonian $H$ be given by
\begin{equation}\label{hami}
H(q^i, p_i, S)=\frac{1}{2} g^{ij}p_i p_j+V(q, S)
\end{equation}
where $(g^{ij})$ is a symmetric bilinear tensor on $Q$. Note that all the integral curves this system satisfies the second law of thermodynamics if and only if
\begin{equation}
	\Delta_{Q}(H) = 2 g^{ij} p_i p_j \geq 0,
\end{equation}
that is, if $g_{ij}$ is a positive semidefinite metric.

This can also be expressed form the brackets defined in~(\ref{eee}). Indeed, we have that 
	\begin{equation}\label{poi}
	\dot{f}=\{f, H\}_{T^*Q}+\{f, H\}_{\Delta_{Q}}.
	\end{equation}
Obviously, $\{H, H\}_{T^*Q}=\{H, H\}_{\Delta_Q}=0$  (first law) and 
$\{S, H\}_{T^*Q}=0$ and $\{S, H\}_{\Delta_Q}=\Delta_Q H\geq 0$ (second law). Observe that in Equation (\ref{poi}) both brackets are using the function $H$ as ``generator''. This is the reason that typically this formalism is known as {\sl single generator formalism}~\cite{Ed-Ber}.

\begin{example}{\bf Linearly damped system}
{\rm 
	
 Consider a linearly damped system~\cite{termo1} described by coordinates  $(q, p, S)$, where $q$ represents the position, $p$ the momentum of the particle and $S$ is the entropy of the surrounding thermal bath. {We assume that the system is subjected to a  viscous friction force, proportional to the minus velocity of the particle. }
The system is described by the  Hamiltonian 
\[
H(q, p, S)=\frac{p^2}{2m}+V(q)+\gamma S, \quad \gamma>0
\]
and $T=\frac{\partial H}{\partial S}=\gamma>0$ represents the temperature of the thermal bath.

Therefore, the equations of motion for ${\mathcal E}_H=\sharp_{\Lambda}(dH)$ are:
	\begin{eqnarray*}
		\dot{q}&=&\frac{p}{m}\\
		\dot{p}&=&-V'(q)-\gamma p\\
		\dot{S}&=& \frac{p^2}{m}
	\end{eqnarray*}

Obviously, the system is isolated since $\dot{H}=0$ and it is also clear from the equation for $\dot{S}$ that the first and second laws are satisfied since $\dot{S}\geq 0$.

In the Lagrangian side we obtain the system given by 
\begin{eqnarray*}
m\ddot{q}&=&{-V'(q)-\gamma m \dot{q}}\\
	\dot{S}&=& m\dot{q}^2.
\end{eqnarray*}
Observe that in this system the friction force is given by the map $F_{fr}:TQ\rightarrow T^{*}Q$ given by
\begin{equation*}
	{F_{fr}(q, \dot{q})=\gamma \dot{q}^{i} dq^{i}.}
\end{equation*}

Therefore, the equation of {entropy production} can be rewritten in terms of the friction force as follows
\[
T\dot{S}= -\langle F_{fr}(q, \dot{q}), \dot{q}\rangle
\]
{These equations coincide with the set of equations proposed in~\cite{Gay-Balmaz2017, Gay-Balmaz2019} for this particular choice of Lagrangian $L$ and friction force $F_{fr}$. Observe that, in this particular example where the temperature satisfies $T=\gamma$, the equations are only defined for values $\gamma>0$ and thus we are only modelling thermodynamic systems with non-zero temperature. }

Observe that the two brackets give
\begin{eqnarray*}
	\{H, g\}_{\Lambda_0}&=&\frac{p}{m}\frac{\partial g}{\partial q}-\frac{\partial g}{\partial p}V'(q)\\
	\{H, g\}_{\Delta_Q}&=&\frac{p^2}{m}\frac{\partial g}{\partial S}-\gamma p\frac{\partial g}{\partial p}
\end{eqnarray*}
and
\[
{\mathcal E}_H(g)=\dot{g}=\{H, g\}_{\Lambda_0}+\{H, g\}_{\Delta_Q}.
\]
Therefore it is clear that $\{H, H\}_{\Lambda_0}=0$ and $\{H, H\}_{\Delta_Q}=0$ (by skew-symmetry) and 
$\{H, S\}_{\Lambda_0}=0$ and 
$\{H, S\}_{\Delta_Q}=\frac{p^2}{m}\geq 0$.
}
\end{example}

\section{Composed thermodynamic systems without friction}\label{sec5}

In this section we will present a model for systems composed of at least two subsystems exchanging heat with each other (see \cite{Romero2010, upm39399}).

Consider two thermodynamic  systems indexed by $1$ and $2$ which may interact through a conducting wall.  On each system we have defined the corresponding Hamiltonian: 
\[
H: T^*(Q_1\times Q_2)\times{\mathbb R}^2\rightarrow {\mathbb R}
\]
where we consider coordinates $(q_\alpha, p_\alpha, S_\alpha)$ on $T^*Q_{\alpha}\times {\mathbb R}$, $\alpha=1, 2$,  where $S_\alpha$ are the entropies of each  subsystem.

In thermo-mechanical systems, as it is usual in Thermodynamics, the partial derivative of the energy with respect to the entropy will be the temperature of the system so that
\begin{equation*}
	T_{\alpha}(q_{1},p_{1},S_{1},q_{2},p_{2},S_{2})=\frac{\partial H}{\partial S_{\alpha}}(q_{1},p_{1},S_{1},q_{2},p_{2},S_{2})
\end{equation*}
denotes the temperature of subsystem $\alpha$.

Consider the Poisson tensor $\Lambda_{Q_1\times Q_2\times {\mathbb R}}$ on $T^*(Q_1\times Q_2)\times {\mathbb R}^2$ given by
\[
\Lambda_{Q_1\times Q_2\times {\mathbb R}}=\Lambda_{Q_1}+\Lambda_{Q_2}+\frac{\partial}{\partial S_{1}}\wedge \frac{\partial}{\partial S_{2}},
\]
where $\Lambda_{Q_\alpha}$ is the canonical Poisson tensor on $T^{*}Q_{i}$ with $\alpha=1,2$.

Assume that both subsystems exchange heat according to Fourier Law:
\[
h=k(T_2-T_1)
\]
where $k$ is the coefficient of thermal conductivity. 
Suppose that $T_\alpha>0$.

Consider the function $K: T^*(Q_1\times Q_2)\times {\mathbb R}^2\rightarrow {\mathbb R}$
\[
K=k \left( \frac{1}{\frac{\partial H}{\partial S_1}}-\frac{1}{\frac{\partial H}{\partial S_2}}\right)=k\left( \frac{1}{T_1}-\frac{1}{T_2}\right)
\]
which will be called \textit{Fourier factor}. Define the two-tensor (with Fourier factor $K$) denoted by $\Lambda_K$, given by
$$
\Lambda_K=\Lambda_{Q_1}+\Lambda_{Q_2}+K \frac{\partial}{\partial S_{1}}\wedge \frac{\partial}{\partial S_{2}}
$$
Observe that now $\Lambda_K$ is a skew-symmetric almost Poisson structure~\cite{MR1091922}.

The matrix representation of $\Lambda_K$ is: 
\[
\left(
\begin{array}{cccccc}
0&I_{{n_1\times n_1}}&0&0&0&0\\
-I_{{n_1\times n_1}}&0&0&0&0&0\\
0&0&0&0&0&k\left( \frac{1}{T_1}-\frac{1}{T_2}\right)\\
0&0&0&0&I_{{n_2\times n_2}}&0\\
0&0&0&-I_{{n_2\times n_2}}&0&0\\
0&0&k\left( \frac{1}{T_2}-\frac{1}{T_1}\right)&0&0&0
\end{array}
\right)
\]
where $\dim Q_i=n_i$, $i=1,2$. 

The  corresponding evolution vector field ${\mathcal E}_{H,K}$:
\begin{equation}
	{\mathcal E}_{H,K}=\sharp_{\Lambda_K}(dH)
\end{equation}	
	The integral curves of  ${\mathcal E}_{H,K}$ are:  
\begin{equation}
	\begin{split}
		\dot{q}_1&=\frac{\partial H}{\partial p_1} \\
		\dot{p}_1&=-\frac{\partial H}{\partial q_1}\\
		\dot{S}_1&=K\frac{\partial H}{\partial S_2}
	\end{split}
	\quad
	\begin{split}
		\dot{q}_2&=\frac{\partial H}{\partial p_2} \\
		\dot{p}_2&=-\frac{\partial H}{\partial q_2} \\
		\dot{S}_2&=-K\frac{\partial H}{\partial S_1}.
	\end{split}
\end{equation}
Observe that the total entropy $S=S_1+S_2$ satisfies
\begin{eqnarray*}
\dot{S}&=&{\mathcal E}_{H,K} (S_1+S_2)\\
&=& k(\frac{T_2}{T_1}-1)+k(\frac{T_1}{T_2}-1)\\
&=&k\frac{(T_2-T_1)^2}{T_1T_2}\geqslant 0
\end{eqnarray*}
	Moreover, in absence of external forces, the total energy $H$ is conserved since by skew-symmetry of $\Lambda_{K}$ we have that
	\[
	{\mathcal E}_{H,K}(H)=0,
	\]
hence the system is isolated.

Since the system presents no friction there is no work done by dissipative forces. Thus, the first law of thermodynamics for this system is just given by
\begin{equation}
    dH = \delta \mathcal{Q} = \delta \mathcal{Q}_1 + \delta \mathcal{Q}_2 = T_1 dS_1 + T_2 dS_2.
\end{equation}
It is easy to comprove that the equality is satisfied by the integral curves of $\mathcal{E}_{H,K}$, since the system is isolated and so the energy is constant along curves.
	
Thus, we have shown that:
\begin{proposition}
	Given a Hamiltonian function $H$, the integral curves of the evolution vector field $\mathcal{E}_{H,K}$ in the skew-symmetric manifold $(T^{*}(Q\times Q\times \R), \Lambda_{K})$, describe the dynamics of an isolated system, that is,
	$$\frac{d H}{d t}=0,$$
	which is composed by two thermodynamic subsystems without friction exchanging heat with each other, satisfying the first and second laws of Thermodynamics, that is
	\begin{equation*}
	    dH = T_1 dS_1 + T_2 dS_2 \quad \text{and} \quad \dot{S} \geqslant 0.
	\end{equation*}
\end{proposition}

	\begin{example}
		The simplest toy model for this case is the two free thermo-particles example, composed by two particles at rest exchanging heat. The thermodynamic phase space is simply $\R^{2}$ on which we define the Hamiltonian function
		\begin{equation*}
			H(S_{a},S_{b})=c_{a}e^{\frac{S_{a}}{c_{a}}}+c_{b}e^{\frac{S_{b}}{c_{b}}},
		\end{equation*}
		where $c_{a},c_{b}$ are the \textit{heat capacities} of each particle. Then the evolution vector field satisfies the equations
		\begin{equation*}
			\begin{split}
				\dot{S}_{a} & =\kappa\left(\frac{T_{b}-T_{a}}{T_{a}}\right)\\
				\dot{S}_{b} & =\kappa\left(\frac{T_{b}-T_{a}}{T_{b}}\right),
			\end{split}
		\end{equation*}
		where the temperatures are the functions given by
		\begin{equation*}
			T_{a}=\frac{\partial H}{\partial S_{a}}, \quad T_{b}=\frac{\partial H}{\partial S_{b}}
		\end{equation*}
		and hence
		\begin{equation*}
			T_{a}=e^{\frac{S_{a}}{c_{a}}}, \quad T_{b}=e^{\frac{S_{b}}{c_{b}}}.
		\end{equation*}
	\end{example}
	
	\begin{example}
		A slight sophistication of of the previous toy model is the two free thermo-particles example, now composed by two particles moving freely on the manifolds $Q_{a},Q_{b}$. The thermodynamic phase space is now given by $(T^{*}Q_{a}\times \R)\times (T^{*}Q_{b} \times \R)$ on which we define the Hamiltonian function
		\begin{equation*}
			H(q_{a},p_{a},S_{a},q_{b},p_{b},S_{b})=\frac{1}{2}\left(\frac{p_{a}^2}{m_{a}}+\frac{p_{b}^2}{m_{b}}\right)+c_{a}e^{\frac{S_{a}}{c_{a}}}+c_{b}e^{\frac{S_{b}}{c_{b}}},
		\end{equation*}
		where $m_{a}, m_{b}$ are the masses of each particle.
	\end{example}	
			
	\begin{example}
		Now, to obtain more interesting examples, we may add to the previous example a potential function depending on the position variables $V:Q_{a}\times Q_{b}\rightarrow \R$. Indeed, a physical example is the two thermo-spring system (cf. \cite{upm39399}).
		
		In this case, the system is modelled by a Hamiltonian function of the type
		\begin{equation*}
			H(q_{a},p_{a},S_{a},q_{b},p_{b},S_{b})=\frac{1}{2}\left(\frac{p_{a}^2}{m_{a}}+\frac{p_{b}^2}{m_{b}}\right)+V(q_{a},q_{b})+c_{a}e^{\frac{S_{a}}{c_{a}}}+c_{b}e^{\frac{S_{b}}{c_{b}}}.
		\end{equation*}	
	\end{example}

\section{Geometric integration of thermodynamic systems}\label{sec6}

Numerical methods for general thermodynamic systems are implemented usually  using the metriplectic formalism (see~\cite{mielke, ignacio}). However, in our case, for the examples that we are considering,  we can easily adapt the construction of discrete gradient methods to the bivector $\Lambda$. 

For  simplicity, we will assume that {$Q=\R^{n}$}. Then the systems that we want to study are described by the  ODEs 
$$\dot{x}=(\sharp_{\Lambda})_x (\nabla H(x)),$$
with {$x=(q^i, p_i, S)\in T^{*}Q\times \R$, the map $H:T^{*}Q\times \R \rightarrow \R$ is the Hamiltonian function and $\nabla H(x)\in\mathfrak{X}(T^{*}Q\times \R)$ is the standard gradient in $T^{*}Q\times \R$ identified with $\R^{2n+1}$}, with respect to the Euclidean metric.

Using discretizations of the gradient $\nabla H(x)$ it is possible to define a class of integrators which preserve the first integral $H$ exactly.

\begin{definition}\label{def31}
	Let $H:\mathbb{R}^N\longrightarrow \mathbb{R}$ be a differentiable function. Then $\bar{\nabla}H:\mathbb{R}^{2N}\longrightarrow \mathbb{R}^N$ is a discrete gradient of $H$ if it is continuous and satisfies
	\begin{subequations}
		\label{discGrad}
		\begin{align}
		\bar{\nabla}H(x,x')^T(x'-x)&=H(x')-H(x)\, , \quad \, \mbox{ for all } x,x' \in\mathbb{R}^N  \, ,\label{discGradEn} \\
		\bar{\nabla}H(x,x)&=\nabla H(x)\, , \quad \quad \quad \quad \mbox{ for all } x \in\mathbb{R}^N  \, . \label{discGradCons}
		\end{align}
	\end{subequations}
\end{definition}

Some examples of discrete gradients are  (see~\cite{MR1694701} and references therein)
\begin{itemize}
	\item The {\bf mean value (or averaged) discrete gradient} given by
	\begin{equation}
	\label{AVF}
	\bar{\nabla}_{1}H(x,x'):=\int_0^1 \nabla H ((1-\xi)x+\xi x')d\xi \, , \quad \mbox{ for } x'\not=x \, .
	\end{equation}
	
	\item The {\bf midpoint (or Gonzalez) discrete gradient} given by
	\begin{align}
	\bar{\nabla}_{2}H(x,x')&:=\nabla H\left( \frac{1}{2}(x'+x)\right)\label{gonzalez}\\&+\frac{H(x')-H(x)
		-\nabla H\left( \frac{1}{2}(x'+x)\right)^T(x'-x)}{|x'-x|^2}(x'-x) \, , \nonumber
	\end{align}
for  $x'\not=x$. 
	\item The {\bf coordinate increment discrete gradient} where  each component given by
	\begin{equation*}
	\label{itoAbe}
	\bar{\nabla}_{3}H(x,x')_i=\frac{H(x'_1,\ldots,x'_i,x_{i+1},\ldots,x_n)-H(x'_1,\ldots,x'_{i-1},x_{i},\ldots,x_n)}{x'_i-x_i}
	\end{equation*}
	$1\leq i \leq N$,  
	when $x_i'\not=x_i$, and $$\bar{\nabla}_{3}H(x,x')_i=\frac{\partial H}{\partial x_i}(x'_1,\ldots,x'_{i-1},x'_i=x_{i},x_{i+1},\ldots,x_n),$$ otherwise.
\end{itemize}

\subsection{Simple thermodynamic systems with friction}

Let $H:T^{*}Q\times \R \rightarrow \R$ be the Hamiltonian function. If we choose the midpoint discrete gradient $\bar{\nabla}{2} H$, it is straightforward to define an energy-preserving integrator by the equation
\begin{equation}\label{discrete:gradient:integrator}
\frac{x_{k+1}-x_k}{h}=\sharp_{\Lambda}\left(\frac{x_k+x_{k+1}}{2}\right)\left(\bar{\nabla}_2H(x_k,x_{k+1})\right),
\end{equation}
where $\Lambda$ is the bivector associated to the canonical contact structure $\eta_{Q}$ of $Q=\R^{2n+1}$, given in local coordinates by \eqref{Lambda:coordinates}.

As in the continuous case, it is immediate to check that $H$ is exactly preserved using \eqref{discrete:gradient:integrator} and the skew-symmetry of $\Lambda$
\begin{equation*}
	\begin{split}
		H(x_{k+1})-H(x_k) & =\bar{\nabla}_2H(x_k,x_{k+1})^T(x_{k+1}-x_k) \\
		& =h \Lambda (\bar{\nabla}_2H(x_k,x_{k+1}),  \bar{\nabla}_2H(x_k,x_{k+1}))=0.
	\end{split}
\end{equation*}

On the other hand, by \eqref{discrete:gradient:integrator} the entropy satisfies
$$
S_{k+1}-S_k=h\Lambda (\bar{\nabla}_2H(x_k,x_{k+1}), dS). 
$$
If $H$ is of the form (\ref{hami}) with $V$ a quadratic function then 
\[
H(x_{k+1})-H(x_k)=dH\left(\frac{x_k+x_{k+1}}{2}\right) (x_{k+1}-x_k).
\]
In fact this is a well-known property of quadratic functions. Hence, we must have
$$d H\left(\frac{x_k+x_{k+1}}{2}\right)=\bar{\nabla}_2 H(x_k,x_{k+1}),$$
so that
\begin{equation*}
	\begin{split}
		S_{k+1}-S_k & = h\Lambda \left(d H\left(\frac{x_k+x_{k+1}}{2}\right), dS\right) \\
		& =h\left(\frac{p_{k}^{i}+p_{k+1}^{i}}{2}\right)\frac{\partial H}{\partial p^{i}}\left(\frac{x_k+x_{k+1}}{2}\right)\geq 0,
	\end{split}
\end{equation*}
since by \eqref{Lambda:coordinates} we have that
\begin{equation*}
	\Lambda(dq^{i},dS)=0, \quad \Lambda(dp_{i},dS)=p_{i} \quad \text{and} \quad \Lambda(dS,dS)=0.
\end{equation*}

\begin{example}
	Consider the Hamiltonian function $H:T^{*}Q\rightarrow \R$ given by
	\begin{equation}\label{harmonic:osc}
	H(q,p,S)=\frac{p^{2}}{2}+\frac{q^{2}}{2}+\gamma S,
	\end{equation}
	where $Q=\R$, which is the Hamiltonian function associated with the damped harmonic oscillator.
	
	Now, if we may apply the midpoint discrete gradient and the associated integrator given by \eqref{discrete:gradient:integrator}, we obtain the following integrator
	\begin{equation}\label{DG2:harmonic}
	\begin{split}
	q_{1} = & \frac{2 \gamma h q_{0}-h^2 q_{0}+4 h p_{0}+4 q_{0}}{2 \gamma h+h^2+4} \\
	p_{1} = & -\frac{2 \gamma h p_{0}+h^2 p_{0}+4 h q_{0}-4 p_{0}}{2 \gamma h+h^2+4} \\
	S_{1} = & \frac{S_{0} h^4+(4 S_{0} \gamma+4 q_{0}^2) h^3+(4 S_{0} 	\gamma^2-16 p_{0} q_{0}+8 S_{0}) h^2}{(2 \gamma h+h^2+4)^2} \\
	& +\frac{(16 S_{0} \gamma+16 p_{0}^2) h+16 S_{0}}{(2 \gamma h+h^2+4)^2}.
	\end{split}
	\end{equation}
	
	Of course, using equations \eqref{DG2:harmonic} we obtain an integrator with constant energy and increasing entropy. In Figure \ref{fig:test1} we can see that the qualitative behaviour of the integrator is fairly accurate, while in Figure \ref{fig:test2} we see the entropy increases at the same rate as the exact one.
	
	\begin{figure}[htb!]
		\centering
		\includegraphics[width=0.7\linewidth]{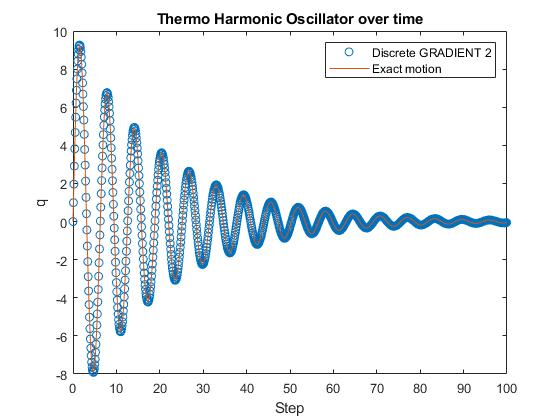}
		\caption{Trajectory of \eqref{DG2:harmonic}: the initial data are $q_{0}=0$, $p_{0}=10$ and $S_{0}=0$; the step is $h=0.1$ and $\gamma=0.1$. We plot the positions $q_{k}$ and compare the integrator with the integral curve of the evolution dynamics ${\mathcal E}_{H}$.}
		\label{fig:test1}
	\end{figure}
	
	\begin{figure}[htb!]
		\centering
		\includegraphics[width=0.7\linewidth]{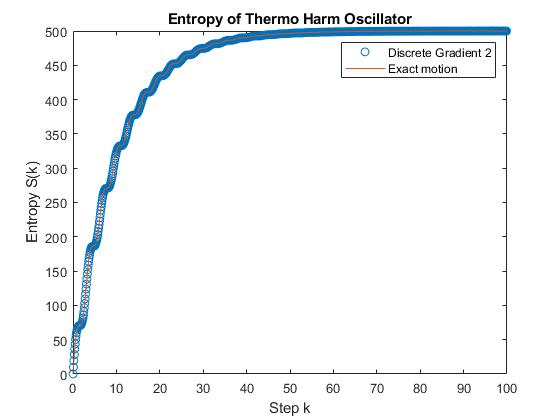}
		\caption{{Entropy} of \eqref{DG2:harmonic}: using the same initial data and settings from Figure \ref{fig:test1}, we plot the error with respect to the exact motion.}
		\label{fig:test2}
	\end{figure}
	
\end{example}

\subsection{Composed thermodynamic systems}

Let $H:T^{*}(Q_{1}\times Q_{2})\times \R^{2}\rightarrow \R$ be the Hamiltonian function defined in a $2(n_{1}+n_{2})+2$ manifold. Moreover, suppose that $Q_{i}$ is a $n_{i}$-dimensional vector space for $i=1,2$.

Using a discrete-gradient approach, given a Hamiltonian function $H:\R^{n}\rightarrow \R$ we may find the midpoint discrete gradient $\nabla_{2}H:\R^{2n}\rightarrow \R^{n}$. Then, in our case, we may define an algorithm in the following way:
\begin{equation*}
\frac{x_{k+1}-x_{k}}{h}=\sharp_{\Lambda_K}\left(\frac{x_{k+1}+x_{k}}{2}\right)(\nabla_{2}H(x_{k},x_{k+1})),
\end{equation*}
where $x_{k}=(q_{1}^{k},p_{1}^{k},S_{1}^{k},q_{2}^{k},p_{2}^{k},S_{2}^{k})$ is a point in $T^{*}(Q_{1}\times T^{*}Q_{2})\times \R^{2}$.

This method will lead to an energy preserving algorithm. Moreover, we have the following result describing the evolution of the total entropy.

\begin{lemma}
	If $H$ is a quadratic function and $T_{i}=\frac{\partial H}{\partial S_{i}}$ we have that
	\begin{equation*}
	S^{k+1}-S^{k}=hk\frac{(T_{2}-T_{1})^2}{T_{1}T_{2}}\geqslant 0,
	\end{equation*}
	where $S^{k}=S_{1}^{k}+S_{2}^{k}$ is the total entropy at step $k$.
\end{lemma}

\begin{proof}
	If $H$ is a quadratic function, the it is not difficult to prove that
	\begin{equation*}
	\nabla_{2}H(x_{k},x_{k+1})=dH\left(\frac{x_{k}+x_{k+1}}{2}\right).
	\end{equation*}
	Moreover, observe that for $i=1,2$
	\begin{equation*}
	S_{i}^{k+1}-S_{i}^{k}=h\Lambda_{K}\left(\nabla_{2}H(x_{k},x_{k+1}),dS_{i}\right).
	\end{equation*}
	Then using the hypothesis that $H$ is a quadratic function, the previous equation becomes
	\begin{equation*}
	S_{i}^{k+1}-S_{i}^{k}=h\Lambda_{K}\left(dH\left(\frac{x_{k}+x_{k+1}}{2}\right),dS_{i}\right)=\pm hK\frac{\partial H}{\partial S_{j}},
	\end{equation*}
	where $j=1,2$ and $j\neq i$. So,
	\begin{equation*}
	\begin{split}
	S^{k+1}-S^{k} & =S_{1}^{k+1}-S_{1}^{k}+S_{2}^{k+1}-S_{2}^{k}\\
	& =hK(-T_{1}+T_{2}) \\
	& =hk\frac{(T_{2}-T_{1})^2}{T_{1}T_{2}}\geqslant 0.
	\end{split}
	\end{equation*}
\end{proof}

Thus, when the Hamiltonian function is a quadratic function, we have a geometric integrator satisfying the first and second laws of Thermodynamics.

\begin{example}
	In the thermo-particle example, described by the Hamiltonian function
	$$H(S_{a},S_{b})=c_{a}e^{\frac{S_{a}}{c_{a}}}+c_{b}e^{\frac{S_{b}}{c_{b}}}$$
	the energy is constant by definition of the discrete gradient function and the skew-symmetry of the structure $\Lambda_K$. Moreover, the total entropy is strictly increasing as it is shown in Figure \ref{Total_entropy_particle_DG} and the temperatures converge to the same value (see Figure \ref{Temperature_particle_DG}).
	
	\begin{figure}[htb!]
		\centering
		\includegraphics[scale=0.4]{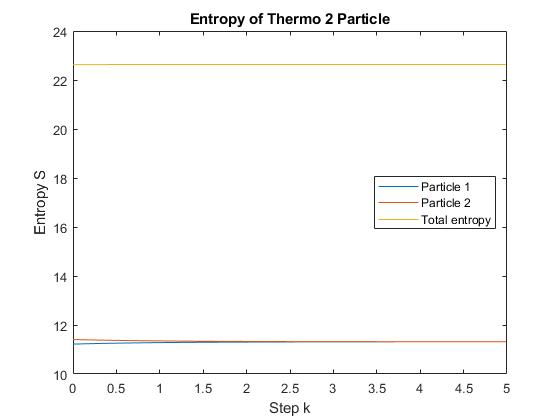}
		\caption{Total entropy of the two thermal particle system. We used $k=1$, $T_1=273.15$, $T_2=300$, $h=0.1$ over $500$ steps.}
		\label{Total_entropy_particle_DG}
	\end{figure}
	
	\begin{figure}[htb!]
		\begin{center}
			\includegraphics[scale=0.4]{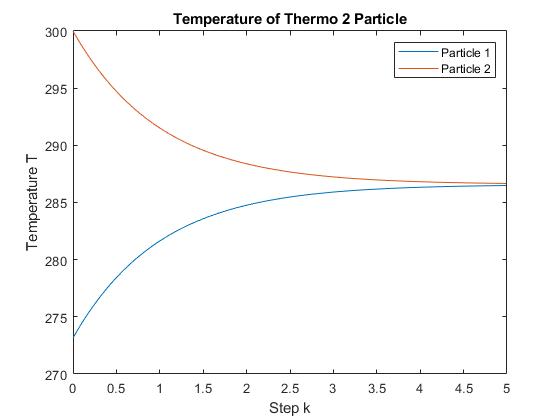}
			\caption{The temperature of each thermal particle in the system. We used $k=1$, $T_1(0)=273.15$, $T_2(0)=300$, $h=0.1$ over $500$ steps.}
			\label{Temperature_particle_DG}
		\end{center}
	\end{figure}
\end{example}

\subsection{``Variational integration'' of the evolution vector field}

Now, we propose   to construct a numerical integrator for ${\mathcal E}_L$  based on a similar method to the discrete Herglotz principle~\cite{AMLL, VBS}. 

Let $L_{d}:Q\times Q\times \R\rightarrow \R$ be a discrete Lagrangian function. Then a possible  integrator for the evolution dynamics is
\begin{equation}\label{DHE}
D_{1}L_{d}(q_{1},q_{2},S_{1})+(1+D_{S}L_{d}((q_{1},q_{2},S_{1}))D_{2}L_{d}(q_{0},q_{1},S_{0})=0
\end{equation}
and the entropy is subjected to
\begin{equation}\label{Herglotz:entropy}
S_{1}-S_{0}=(q_{1}-q_{0})D_{2}L_{d}(q_{0},q_{1},S_{0}).
\end{equation}

\begin{example}
	Consider again the Hamiltonian function \eqref{harmonic:osc} of the damped harmonic oscillator. Since $H$ is regular, we may consider the corresponding Lagrangian function $L:TQ\times \R\rightarrow \R$ given by
	\begin{equation*}
	L(q,\dot{q},S)=\frac{\dot{q}^{2}}{2}-\frac{q^{2}}{2}-\gamma S.
	\end{equation*}
	A standard discretization of this Lagrangian function is given by means of a quadrature rule like
	\begin{equation*}
	L_{d}(q_{0},q_{1},S_{0})=\frac{(q_{1}-q_{0})^{2}}{2h}-h\frac{(q_{1}+q_{0})^{2}}{8}-h \gamma S_{0}.
	\end{equation*}
	The discrete Herglotz equations \eqref{DHE} together with \eqref{Herglotz:entropy} give the explicit integrator
	\begin{equation}\label{Herglotz:harmonic}
	\begin{split}
	& q_{2} = \frac{\gamma h^3 q_{0}+\gamma h^3 q_{1}+4 \gamma h q_{0}-4 \gamma h q_{1}-h^2 q_{0}-2 h^2 q_{1}-4 q_{0}+8 q_{1}}{h^2+4} \\
	& S_{1} = S_{0} + \frac{(q_{1}-q_{0})^2}{h}-h \frac{q_{1}^{2}-q_{0}^{2}}{4}.
	\end{split}
	\end{equation}
	
	In Figures \ref{fig:test3} we plot the integrator given by equations \eqref{Herglotz:harmonic}. We see that the qualitative behaviour of the integrator is also quite good. In fact, an open question is whether the error can be improved by considering discrete Lagrangian functions approximating well enough the exact discrete Lagrangian function.
	
	As a last comment, the entropy for equations \eqref{Herglotz:harmonic} is increasing and the Hamiltonian oscillates before stabilizing around a constant value (cf. Fig \ref{fig:test5}).

	\begin{figure}[htb!]
		\centering
		\includegraphics[width=0.7\linewidth]{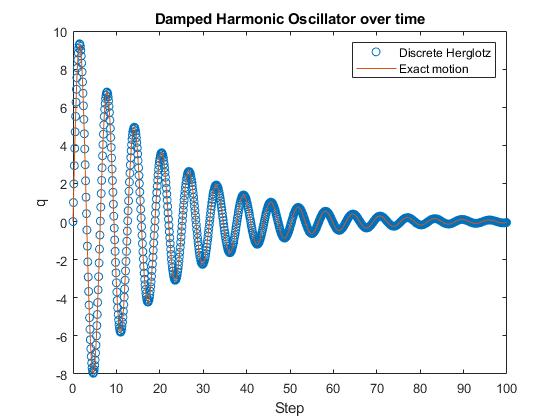}
		\caption{Trajectory of \eqref{Herglotz:harmonic}: the initial data are $q_{0}=0$, $q_{1}=1$ and $S_{0}=0$; the step is $h=0.1$ and $\gamma=0.1$. We plot the positions $q_{k}$ and compare the integrator with the integral curve of the evolution dynamics ${\Gamma}_{L}$.}
		\label{fig:test3}
	\end{figure}
	
	\begin{figure}[htb!]
		\centering
		\includegraphics[width=0.7\linewidth]{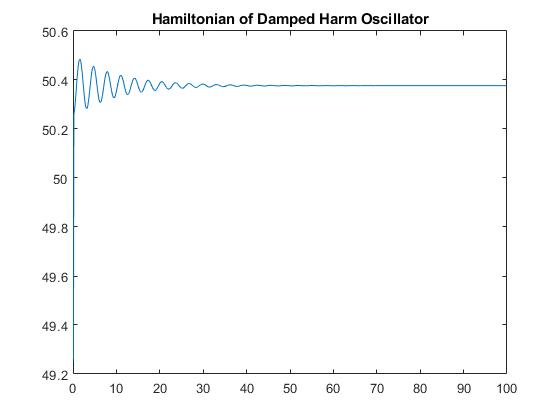}
		\caption{Hamiltonian of \eqref{Herglotz:harmonic}: using the same initial data and settings from Figure \ref{fig:test3}, we plot the Hamiltonian function along the iterations of the integrator.}
		\label{fig:test5}
	\end{figure}
	
\end{example}

\section{Conclusions and future work}\label{sec7}




In this paper, we have given a variational interpretation of the evolution vector field and clarified its relationship with the laws of thermodynamics on isolated systems. Furthermore, we have extended our theory to deal with composed thermodynamic systems without friction and we also provided geometric integrators for this framework.

Nevertheless, there are still many open questions. Related to simple systems, we would like to see if a similar formalism holds for open systems, in which the number of particles and the chemical potentials have to be taken into account, and for closed non-isolated systems, in which the energy is not necessarily constant, such as the ones experimenting isobaric or isothermal processes.

In addition, Hamiltonians of the form~\eqref{hami} where $g_{ij}$ is a Lorentzian metric could have interesting applications to relativity, in particular, in black hole thermodynamics. Though it is true that the second law of thermodynamics will not hold for all integral curves of the evolution vector field associated with a semi-Riemannian metric, it does hold for time-like curves which are the ones that describe the allowed trajectories for matter.

In what concerns multi-component systems, there is still much work to do. Indeed, a Lagrangian formulation and a nonholonomic constraint of the type $\delta \mathcal{Q} = 0$ could be used to introduce a ``variational principle'' from which we derive the integral curves of the evolution vector field $\mathcal{E}_{H,K}$. Also, one could introduce a similar formulation in order to account for friction. It could be interesting to formulate it in such a way one could model thermo-visco-elastic systems with it (cf, \cite{upm39399}). Finally, it is reasonable to think that our formalism could be generalized without much effort to the case of $N$ subsystems exchanging heat with each other. In order to accomplish this, we must consider a skew-symmetric tensor $\Lambda_K$, encoding in component function $K_{ij}$ the heat flux interchanged by particles $i$ and $j$, with $i\neq j \in \{1,...,N\}$.

\section*{Acknowledgements}
The authors acknowledge financial support from the Spanish Ministry of Science and Innovation, under grants PID2019-106715GB-C21, MTM2016-76702-P, ``Severo Ochoa Programme for Centres of Excellence in R\&D'' (CEX2019-000904-S) and from the Spanish National Research Council, through the ``Ayuda extraordinaria a Centros de Excelencia Severo Ochoa'' (20205\-CEX001). A. Simoes is supported by the FCT (Portugal) research fellowship SFRH/BD/129882/2017.

\bibliography{Thermo_Discrete}
\bibliographystyle{plain}
\end{document}